\documentclass{amsart}
\usepackage[T1]{fontenc}
\usepackage[utf8]{inputenc}

\usepackage[
backend=biber,
style=alphabetic,
sorting=nyt
]{biblatex}
\newtheorem{theorem}{Theorem}
\newtheorem{lemma}{Lemma}
\newtheorem{corollary}{Corollary}
\theoremstyle{definition}
\newtheorem{Definition}{Definition}
\theoremstyle{remark}
\newtheorem{example}{Example}

\newenvironment{Example}{\begin{example}}{\end{example}}

\usepackage[ruled,noend,linesnumbered]{algorithm2e}
\SetKwInOut{KwIn}{Input}
\SetKwInOut{KwOut}{Output}

\author{Alaa Ibrahim\and Bruno Salvy}
\email{Alaa.Ibrahim@inria.fr, Bruno.Salvy@inria.fr}
\address{Univ Lyon, EnsL, UCBL, CNRS, Inria,  LIP, F-69342, LYON Cedex 07, France}

\usepackage{todonotes}

\usepackage{graphicx}

\newcommand{\trsp}[1]{#1^\mathsf{T}} 

\def\PP{\textsf{PositivityProof}}

\usepackage[colorlinks=true,linkcolor=blue,citecolor=blue,urlcolor=magenta,pdfpagelabels=false]{hyperref}

\usepackage[capitalise]{cleveref}
\crefname{@theorem}{Theorem}{Theorems}

\title{Positivity Certificates for Linear Recurrences}  

\begin{document}

\begin{abstract} 
We consider linear recurrences with polynomial coefficients of Poincaré type and with a unique simple dominant eigenvalue. We give an algorithm that proves or disproves positivity of solutions provided the initial conditions satisfy a precisely defined genericity condition. For positive sequences, the algorithm produces a certificate of positivity that is a data-structure for a proof by induction. This induction works by showing that an explicitly computed cone is contracted by the iteration of the recurrence.
\end{abstract}
\maketitle

\section{Introduction}
A sequence $(u_n)_{n\in \mathbb{N}}$ of real numbers is called \emph{P-finite} if it satisfies a  linear recurrence 
 \begin{equation}\label{rec}
    p_d(n)u_{n+d}=p_{d-1}(n) u_{n+d-1}+\dots+p_0(n)u_n,\qquad n\in \mathbb{N},
 \end{equation}
with coefficients $p_i\in\mathbb R[n]$~(\footnote{Other names for such sequences are \emph{P-recursive}~\cite{Stanley1980} and \emph{holonomic}. The name P-finite was introduced by Zeilberger~\cite{Zeilberger1990}. It is more consistent with the use of `C-finite' for constant coefficients and `D-finite' for linear differential equations. It is also the choice made in recent works by Kauers and Pillwein~\cite{Kauers2006,KauersPillwein2010}.}). When the coefficients~$p_i$ are constants in~$\mathbb R$, the sequence is called \emph{C-finite}.
If $p_d\neq 0$, the \emph{order} of the relation \eqref{rec} is $d$. If $0\not\in p_d(\mathbb N)$, then the sequence is completely determined by the recurrence and initial conditions $(u_0,\dots,u_{d-1})$. We  make this assumption in the rest of this article. (\footnote{\label{footnote2}When it does not hold, the sequence can be defined with extra initial conditions $u_{i+d}$ for $i$ s.t. $p_d(i)=0$. For positivity questions, dealing with~$k:=\max(i\in\mathbb N\mid p_d(i)=0)$ initial values of the sequence separately and considering the recurrence satisfied by~$(u_{n-k})_{n\in\mathbb N}$ reduces to the situation when~$0\not\in p_d(\mathbb N)$.})

Given the polynomials $p_i$ and initial conditions,
the \emph{positivity problem} is to decide whether $u_n\ge0$ for
all~$n\in\mathbb N$~(\footnote{We call this a problem of positivity rather than non-negativity to be consistent with the terminology used in the C-finite case~\cite{OuaknineWorrell2014a}. The related problem to decide whether $u_n>0$ for all~$n\in\mathbb N$ is also of interest; our results carry over to that case.}).
For instance, the rational sequence
\begin{equation}\label{eq:Straub}s_n=\sum_{k=0}^n{(-27)^{n-k}2^{2k-n}\frac{(3k)!}{k!^3}}\binom{k}{n-k}
\end{equation}
is not obviously positive. One way of proving its positivity starts from the recurrence
\[2(n+2)^2s_{n+2}=(81n^2+243n+186)s_{n+1}-81(3n+2)(3n+4)s_n,\quad s_0=1,s_1=12\]
that can be computed by Zeilberger's algorithm~\cite{PetkovsekWilfZeilberger1996}.
A general method due to Gerhold and Kauers~\cite{GerholdKauers2005}, turned into an algorithm for recurrences of order~2 by Kauers and Pillwein~\cite{KauersPillwein2010}, shows the positivity of the solution. 
(Another proof of the positivity of~$s_n$ was given by Straub and Zudilin using hypergeometric identities
\cite{StraubZudilin2015}.) In this work, we give an algorithm proving positivity of a large class of sequences of arbitrary order, including those dealt with by the algorithm of Kauers and Pillwein.

\medskip
P-finite and C-finite sequences are
closed under addition, product and Cauchy product
\(((u_n)_{n\in\mathbb N},(v_n)_{n\in\mathbb N})\mapsto(\sum_{k=0}^nu_kv_{n-k})_{n\in\mathbb N}\). Also, for any
\(\ell\in\mathbb N_{>0}\) and \(q\in\{0,\dots,\ell-1\}\), the
subsequence \((u_{\ell n+q})_{n\in\mathbb N}\) satisfies a linear recurrence (of order
at most \(d\)). These operations are all effective, so that recurrences can be computed for these sequences given recurrences for the input~\cite{Stanley1999}.
These closure properties allow to reduce other problems to that of positivity.

\begin{Example} If \((u_n)_{n\in\mathbb{N}}\) is a C-finite sequence of rational
numbers and \(m\) is the lcm of
the denominators of the initial conditions \(u_0,\dots,u_{d-1}\) and
of the
coefficients \(p_0,\dots,p_{d-1}\), then the sequence defined by 
$v_n=m^nu_n$ is a C-finite sequence of integers and \((w_n)_{n\in
\mathbb{N}}=(v_n^2-1)_{n\in\mathbb{N}}\) is another C-finite
sequence of integers, which is positive if and only if \(u_n\neq0\) for all \(n\). Thus Skolem's problem, which is notoriously difficult, reduces to positivity, thereby showing that positivity is also likely to be hard~\cite{HalavaHarjuHirvensalo2006,OuaknineWorrell2014b}.
\end{Example}

\begin{Example} Deciding whether \(u_n\ge v_n\) for all $n\in\mathbb N$ reduces to the positivity of \((u_n-v_n)_{n\in\mathbb N}\).
Similarly, deciding that \((u_n)_{n\in\mathbb N}\) is increasing (\(u_{n+1}\ge u_n\) for
all \(n
\)), or convex (\(u_{n+1}+u_{n-1}\ge 2u_n\)) or log-convex
(\(u_{n+1}u_{n-1}\ge u_n^2\)) all reduce to the positivity problem, by
constructing recurrences for these new sequences.
\end{Example}

For applications of the positivity problem of C-finite sequences, we refer to the numerous references in the work of Ouaknine and Worrell~\cite{OuaknineWorrell2014a}. 
Motivations for studying positivity in the more general context of
P-finite sequences also come from various areas of mathematics and its
applications, including number theory~\cite{StraubZudilin2015},
combinatorics~\cite{ScottSokal2014}, special function theory~
\cite{Pillwein2008}, or biology~\cite{MelczerMezzarobba2022}. In
computer science, the verification of loops allowing multiplication by
the loop counter leads to P-finite sequences~%
\cite{HumenbergerJaroschekKovacs2017,HumenbergerJaroschekKovacs2018}.
Positivity questions for such recurrences also occur in the
floating-point
error analysis of simple loops obtained by discretization of linear
differential equations~%
\cite{BoldoClementFilliatreMayeroMelquiondWeis2014} and in the
numerical stability of the computation of sums of convergent power
series~\cite{SerraArzelierJoldesLasserreRondepierreSalvy2016}.

\medskip
\paragraph{\bf Previous works}
For C-finite sequences of rational numbers, Ouaknine and Worrell have shown decidability of positivity for recurrences of order up to~5, and a relation between the decidability in higher order and the computability of the homogeneous Diophantine approximation of a specific set of transcendental numbers, a problem related to difficult questions in analytic number theory~\cite{OuaknineWorrell2014a}. We refer to their work for earlier references. When the characteristic polynomial of the sequence does not have multiple roots, this extends to order up to~9. For reversible recurrences of integers (reversible means that unrolling the recurrence backwards produces only integers for negative indices), decidability of positivity is known for order up to~11 and this goes up to~17 if the recurrence is both reversible and with square-free characteristic polynomial~\cite{KenisonNieuwveldOuaknineWorrell2023}. Closer to our work, for recurrences having one dominant eigenvalue, decidability is proven for arbitrary order \cite{OuaknineWorrell2014b}. This is the property we use for P-finite sequences.

For P-finite sequences of order~1, positivity is easy. For order~2, it is reducible to the problem of \emph{minimality}~\cite{KenisonKlurmanLefaucheuxLucaMoreeOuaknineWhitelandWorrell2021}, itself a special case of \emph{genericity of initial conditions} that appears in our work.

Another approach to the positivity of P-finite sequences starts with the work of Gerhold and Kauers~\cite{GerholdKauers2005}, who suggest to check for
increasingly large \(k\) whether
\[u_n\ge0\wedge u_{n+1}\ge0\wedge\dots\wedge u_{n+k}\ge0\Rightarrow u_{n+k+1}\ge0.\]
Using the recurrence, this can be rewritten as a decision problem in
the existential theory of the reals. This can be solved by
cylindrical algebraic decomposition~\cite{collins1975quantifier},
which is what they use; other approaches based on critical points are
also possible~\cite[ch.~13]{BasuPollackRoy2008}.

Gerhold and Kauers obtained several successes with their method,
notably an
automatic proof of Tur\'an's inequality for Legendre polynomials,
\[P_n(x)^2-P_{n-1}(x)P_{n+1}(x)\ge0,\qquad x\in[-1,1],\]
that involves a parameter~\cite{GerholdKauers2006}. But termination is not guaranteed in general and sufficient conditions
for the success of this method are unclear~\cite{KauersPillwein2010}.

Kauers and Pillwein focused on the application of
this method to P-finite sequences~\cite{KauersPillwein2010}. They
added the idea of looking for
a proof by induction of the inequalities \(u_{n+1}\ge\mu u_n\ge 0\)
for a
well-chosen real \(\mu>0\). They showed that this works for order~2
with generic initial conditions. They isolated a class
of recurrences of order~3 for which this approach
also works. Pillwein~\cite{Pillwein2013} explored variants of this method and extended the class of recurrences that can be handled with this type of method. Recently, Pei, Wang, Wang~\cite{PeiWangWang2023} revisited the case of order~2 and gave a simple way to compute
\(\mu\) as above, and \(N\) such that \(u_{n+1}\ge\mu u_n>0\) for
\(n\ge N\).

\medskip
\paragraph{\bf Contributions} Our starting point is a result of
Friedland on the convergence of products of the successive elements of a convergent sequence of matrices~\cite{Friedland2006}. We make explicit the effective aspects of some
of his proofs and apply them to questions of positivity. 
We deal with P-finite sequences of Poincaré type, which means that after dividing by the leading coefficient and taking the limit $n\rightarrow\infty$, each of the coefficients has a finite limit. (We show in \cref{sec:dominant-eigenvals} how to reduce to this case.)
Moreover, we demand that the characteristic polynomial of this new
recurrence has only one root of maximum modulus and that it is a
simple root. Then, we show that, except for a hyperplane of initial
conditions, positivity can be proved by an induction that proves
$d(d-1)$ linear inequalities simultaneously. Note that for order $d=2$, \(d(d-1)=2\) is also the number of inequalities used by Kauers and Pillwein. 

These inequalities have a geometric nature: they describe a convex cone containing the vector~$(u_n,u_{n+1},\dots,u_{n+d-1})$, bounded by $d(d-1)$ hyperplanes and contained in~$\mathbb R_{>0}^d$. The proof by induction consists in proving that successive vectors do not leave that cone. Our algorithm thus produces that cone and an integer~$N$ such that at index~$N$, the vector has entered the cone and no~$u_n$ of smaller index is negative. Capturing the geometry of the iteration by means of over-approximations by cones or related geometric surfaces is natural in this context. For the less general C-finite case and more general questions than positivity, related (but distinct) surfaces have been used  recently~\cite{AlmagorChistikovOuaknineWorrell2022}.

Like Friedland's result, our approach applies to the situation
of a linear recurrence~$U_{n+1}=A(n)U_n$, where~$A(n)$ is a square
matrix over~$\mathbb R(n)$ that is invertible for
all~$n\ge0$ and whose limit as~$n\rightarrow\infty$ is finite. For
positivity, we further require that the limit has a unique eigenvalue
of maximal modulus that is simple, and a corresponding eigenvector
with positive coordinates.

This work is structured as follows. First, background on eigenvalues and asymptotics of linear recurrences is recalled in \cref{sec:background}. \Cref{sec:certif} presents our result, the positivity certificates and how they are verified. The ideas leading to the algorithm are presented in \cref{sec:Friedland}, where we describe the relevant tools from Friedland's work. The algorithm is then given with its proof in \cref{sec:algo}. 

\section{Background}\label{sec:background}
\subsection{Algebraic coefficients}
P-recursivity can be defined over arbitrary fields, but as we are interested in
positivity issues, it is natural to restrict our attention to
subfields of~$\mathbb R$. More precisely, we denote by~$\mathbf Q$ a
field that is either the field~$\mathbb Q$ of rational numbers, or a
real number field~$\mathbb Q(\alpha)$, where~$\alpha$ is given, for
instance, by a square-free polynomial and an isolating 
interval~\cite{BasuPollackRoy2008,Yap2000}. 
In particular, with this
data structure,
it is
possible to determine the sign of an
element
of~$\mathbf Q$, where `sign' means any of $<0$ or $>0$ or $=0$. From
there, using Sturm sequences, one can compute the
number of roots of a polynomial in~$\mathbf Q[x]$ in an interval with
endpoints that are either infinite or in~$\mathbf Q$. A direct consequence used repeatedly in this
work is that one can determine an integer beyond which a polynomial
in~$\mathbf Q[x]$ has fixed sign. (For this problem, one can also use simple Cauchy-type bounds~\cite[Thm.~4.2]{Mignotte1992a}.) In some cases, we also use the fact
that these algorithms extend to~$\mathbf Q(\lambda)$
with~$\lambda\in\mathbb R$
algebraic over~$\mathbf Q$.

\subsection{Dominant eigenvalues}\label{sec:dominant-eigenvals}
If $U_n$ denotes the vector $\trsp{(u_{n},\dots,u_{n+d-1})}$,  the
linear recurrence~\eqref{rec} of order~$d$ is a special case of a first-order linear recurrence 
\begin{equation}\label{eq:rec-vector}
U_{n+1}=A(n)U_n,
\end{equation}
where $A(n)\in\mathbf Q(n)^{d\times d}$. In the situation of \cref{rec}, $A(n)$ is the companion matrix
\[A(n)=\begin{pmatrix}
    0 & 1 & 0 & \dots& 0\\
    0 & 0 & 1 & \dots& 0\\
    \dots&\dots &\dots&\dots&\dots\\
    0 &0  &0 &\dots&1\\
    \frac{p_0(n)}{p_d(n)}&\frac{p_1(n)}{p_d(n)} & \frac{p_2(n)}{p_d(n)}&\dots&\frac{p_{d-1}(n)}{p_d(n)}
    \end{pmatrix}, \qquad n\in\mathbb{N}.\]
The sequence $U_n$ is then recovered from the vector of initial conditions by the matrix factorial $U_n=A(n-1)A(n-2)\dotsb A(0)U_0$.

\begin{Definition}
The linear recurrence~\eqref{eq:rec-vector} (and also~\eqref{rec} as a
special case) is said to be of \emph{Poincaré type} if the matrix
$A:=\lim_{n\rightarrow\infty}A(n)$ is finite (i.e., all the entries
of~$A(n)$ have a finite limit).
\end{Definition}
The motivation for considering this notion is that the finite case corresponds to the situation of a linear recurrence with constant coefficients. Then the P-finite case can be viewed as a 
perturbation of the C-finite case.

For linear recurrences of the type of~\cref{rec}, being of Poincaré type is not a strong restriction for positivity
questions. The general case can be reduced to the Poincaré type~\cite[\S2]{MezzarobbaSalvy2010}. In summary, if the recurrence is not of Poincaré type, then one of
the~$p_i$ has degree higher than that of~$p_d$ and a
solution behaves asymptotically like a rational power~$p/q$ of~$n!$.
The maximal such power can be found by a Newton polygon (this observation goes back to Perron and Kreuser).
Then, one can consider the P-finite
sequence obtained by multiplying $u_n$ by the solution of~$n^pu_
{n+q}=u_n$, with initial conditions~$(1,\dots,1)$. The same operation
can also be used if the matrix~$A$ is nilpotent, using a recurrence of
the form~$u_{n+q}=n^pu_n$ instead, so that we can always assume
that the recurrence is of Poincaré type, with~$A$ having a nonzero
eigenvalue.

\begin{Definition}[Dominant eigenvalues]Let
$\lambda_1,\dots,\lambda_m$ be the distinct complex eigenvalues of the
limit matrix $A$, numbered by decreasing modulus so that 
\[|\lambda_1|=|\lambda_2|=\dots=|\lambda_k|>|\lambda_{k+1}|\geq |\lambda_{k+2}|\dots \geq |\lambda_m|.\]
Then $\lambda_1,\dots,\lambda_k$ are called the \emph{dominant eigenvalues of $A$} (or equivalently a \emph{dominant root} of its characteristic polynomial). We say that an eigenvalue is \emph{simple} when it is a simple root of the characteristic polynomial.
\end{Definition}
Given a characteristic polynomial in~$\mathbf Q[x]$, one can isolate
the dominant eigenvalues in polynomial bit
complexity~\cite{GourdonSalvy1996}, see also~\cite{BugeaudDujellaFangPejkovicSalvy2022}.

\subsection{Asymptotics}
For C-finite sequences, a starting point is the closed form
\[u_n=\sum_{i=1}^kC_i(n)\lambda_i^n+\sum_{i>k}C_i(n)\lambda_i^n,\]
split into one sum over dominant eigenvalues and one sum over smaller ones. Since the basis of solutions is known explicitly, the polynomials~$C_i(n)$ can be computed easily from the initial conditions. They belong to $\mathbf Q(\lambda_1,\dots,\lambda_m)[n]$.

The difficulty when using this formula to prove positivity is that for
$k>1$, the first sum contains oscillating sequences that can come very
close to~0. This is where tools from analytic number theory, such as
Baker's theorem on linear forms in logarithms, come into play for
deciding positivity~\cite{kenison_et_al:LIPIcs.ICALP.2023.130}.

For P-finite sequences the situation is made harder by the fact
that there is no `simple' basis of solutions. Also, the constants that
appear (the analogues of the coefficients of the $C_i$ above), even in
the leading coefficient of the asymptotic behaviour, are difficult to
relate to
the initial conditions. This is illustrated by the following.
\begin{Example}
The number of `fragmented permutations' of size~$n$ is~$c_n/n!$ where $(c_n)_{n\in\mathbb N}$ is defined by
\[(n+2)c_{n+2}=(2n+3)c_{n+1}-nc_n,\quad c_0=c_1=1.\]
It satisfies~\cite[Prop.~VIII.4]{FlajoletSedgewick2009}
\[c_n
\sim\frac{n^{-3/4}e^{2\sqrt n}}{2\sqrt{e\pi}},\quad n\rightarrow\infty.\]
Here, the leading coefficient $1/2\sqrt{e\pi}$ is computed (with the rest of the asymptotic behaviour) by exploiting a closed-form expression of the generating function of~$(c_n/n!)_{n\in\mathbb N}$. In general, this is not available.
\end{Example}
While we know how to compute a basis of formal asymptotic expansions that are solutions of linear recurrences (by the results of Birkhoff-Trjitzinsky improved by Immink~\cite{BirkhoffTrjitzinsky1932,Immink1984}), 
we do not know how to compute the leading coefficient exactly in general. Currently, the closest we have is a certified numerical approximation in the form of an interval that can be made arbitrarily small, but~0 cannot be excluded. This is known as the \emph{connexion problem} for linear differential equations. Still, this is a good basis for an analytic proof of positivity, as was recently shown by Melczer and Mezzarobba on a recurrence of order~7 with polynomial coefficients, themselves of degree~7~\cite{MelczerMezzarobba2022,DongMelczerMezzarobba2022}.

In the case of Poincaré-type recurrences, Poincaré related the asymptotic behaviour to the C-finite case, showing that when all the eigenvalues are  simple and of distinct moduli, any solution $(u_n)_{n\in\mathbb N}$ of the recurrence is either  ultimately~0 (all its terms are~0 from a certain index on) or satisfies $\lim_{n\rightarrow\infty} u_{n+1}/u_n=\lambda_i$ for some~$i$~\cite{Poincare1885}.
Solutions that are ultimately~0 exist if and only if the trailing coefficient~$p_0(n)$ of~\cref{rec} vanishes at a positive integer, or equivalently, when the matrix~$A(n)$ is not invertible for some~$n\in\mathbb N$. For positivity testing, one can proceed as for the leading coefficient: treat the initial terms of the sequence separately up to the largest integer where this happens, and then shift the index.

Perron, Kreuser and later Kooman gave results of a converse type:
 sufficient conditions for a solution to exist with limit~$u_
 {n+1}/u_n=\lambda_i$~\cite{Kooman1991}.
We rely on the following more recent analytic result, which we will use with sequences of invertible matrices with entries in~$\mathbf Q$.
\begin{theorem}[Friedland~\cite{Friedland2006}]\label{thm:Friedland}
Let~$A(n)$ be in~$\operatorname{GL}_d(\mathbb C)$ for~$n\in\mathbb N$ and tend to a finite limit~$A$ as $n\rightarrow\infty$, such that $A$ has exactly one dominant eigenvalue $\lambda$. 
Then there exist two nonzero vectors $v,w$ and a sequence of real numbers~$\theta_n$ such that $Av=\lambda v$
and
\[
\lim_{n\rightarrow\infty}e^{i\theta_n}\frac{A(n)\dotsm A(1) A(0)}{\|A(n)\dotsm A(1) A(0)\|}=v\trsp{w}.\]
A vector of initial conditions $U_0$ is called \emph{generic} when $\trsp{w}U_0\neq0$.
\end{theorem}
Thus, for a generic vector of initial conditions, the
sequence $U_n$, which equals $A(n)\dotsm A(0)U_0$, has a direction that tends to that of~$v$, in a sense
made more precise in \cref{sec:Friedland}.

In the constant case, when $A(n)=A$ for all~$n$, this theorem gives a
proof of the convergence of the classical power method~%
\cite{ParlettPoole1973}. In that situation, the vector $w$ is a left
eigenvector of~$A$ for~$\lambda$, i.e., $\trsp{w}A=\lambda \trsp{w}$.
In particular, if the entries of~$A$ belong to the
field of rational numbers~$\mathbb Q$, then the entries of~$w$
belong to~$\mathbb Q(\lambda)$. In the case of polynomial
coefficients, this vector~$w$ is
much more elusive.
\begin{Example}
The recurrence used by Apéry in his proof of the irrationality of~$\zeta(3)$ \cite{Van-der-Poorten1979}
 is
\[(n + 2)^3u_{n + 2} =(2n + 3)(17n^2 + 51n + 39)u_{n + 1} -(n+1)^3u_n.\]
The corresponding limit matrix has eigenvalues $\lambda_\pm=
(3\pm2\sqrt2)^2$, and corresponding
eigenvectors $\trsp{(1,\lambda_\pm)}$. Up to a nonzero scalar, the vector $\trsp{w}$ is $(1,6/\zeta(3)-5).$ Since~$\zeta(3)$ is irrational, any nonzero vector of initial conditions in~$\mathbb Q$ is generic.
\end{Example}
In order 2, the non-generic situation is called \emph{minimal} as it corresponds to a vector space of dimension~1 of solutions that do not have the dominant order of growth. For recurrences of order~2, deciding positivity reduces to deciding minimality~\cite{KenisonKlurmanLefaucheuxLucaMoreeOuaknineWhitelandWorrell2021}. \Cref{theorem:main} below generalizes this situation to arbitrary order.

\section{Positivity certificates}\label{sec:certif}

\begin{Definition}
We say that a real vector or matrix $V$ is \emph{positive} 
(resp. non-negative), and write $V>0$ (resp. $V\ge0$), when all its
entries are positive (resp. non-negative).
\end{Definition}

For generic initial conditions, it is a consequence of Pringsheim's theorem~\cite{Titchmarsh1939} (see also~\cite{Vivanti1893a,Hadamard1954}) that if $(u_n)_{n\in\mathbb N}$ is a positive solution of~\cref{rec}, and $\lambda$ is a dominant eigenvalue of the limit matrix~$A$, then $|\lambda|$ itself is an eigenvalue of it. Thus if there is a unique dominant eigenvalue~$\lambda$, it is real and positive. Moreover, if $\lambda$ is an eigenvalue of a companion matrix, then~$
(1,\lambda,\dots,\lambda^{d-1})$ is a corresponding \emph{positive} eigenvector. For more general recurrences \cref{eq:rec-vector} with $A(n)$ an arbitrary matrix of rational functions, we add the existence of a positive eigenvector as a hypothesis in our approach.
Our main result is the following.
\begin{theorem}\label{theorem:main}
Let~$A(n)$ be in $\mathbf Q(n)^{d\times d}$, invertible for $n\in\mathbb N$, and tending to a finite limit~$A$ as~$n\rightarrow\infty$, that has a unique simple dominant eigenvalue and a corresponding positive eigenvector. Then there exists a vector $W\in\mathbb R^d$ such that  positivity of the solution of~$U_{n+1}=A(n)U_n$ given~$A(n)$ and~$U_0\in\mathbf Q^d$ can be decided when $\trsp{W}U_0\neq0$.

Algorithm \PP\ in \cref{sec:algo} either disproves positivity or computes a positivity certificate, in the generic situation $\trsp{W}U_0\neq0$.
\end{theorem}

By the discussion above, we obtain the following consequence for P-finite sequences.
\begin{corollary}\label{cor:recurrences}
Given the polynomials~$p_i\in\mathbf Q[n]$ in the linear recurrence \cref{rec} 
and initial conditions $U_0=(u_0,\dots,u_{d-1})\in\mathbf Q^d$, if $0\not\in p_0p_d(\mathbb N)$, $\deg p_d=\deg p_0\ge \deg p_i$ for $i\in\{1,\dots,d-1\}$ and the characteristic polynomial
\[\chi(X)=X^d-\sum_{i=0}^{d-1}{\lim_{n\rightarrow\infty}\frac{p_i(n)}{p_d(n)}X^i}\]
has a unique dominant root, then there exists a nonzero~$W\in\mathbb R^d$ such that the positivity of the sequence $(u_n)$ can be decided if $\trsp{W}U_0\neq0$.
\end{corollary}
The cases when $0\in p_0p_d(\mathbb N)$ can be handled as in \cref{footnote2}.

If both the initial condition is not generic and the sequence is positive, then, and only then, our algorithm does not terminate. Constructing examples of minimal-order recurrences with coefficients in~$\mathbb Q[n]$ and initial conditions in~$\mathbb Q$ where this occurs does not seem to be easy.

\subsection{Certificates and their verification} 
In \cref{theorem:main}, a positivity certificate is a data-structure 
for a proof by induction: it consists of a quadruple $(T,r,N,m)$
formed of an invertible matrix $T\in\operatorname{{GL}}_d(\mathbb Q)$,
a rational number $r>1\in\mathbb Q$ (or $r=\infty$), a
non-negative integer $N\in\mathbb N$, and a positive integer
$m\in\mathbb N_{>0}$\footnote{This last integer~$m$ is there for technical reasons and does not have a geometric meaning; we suggest focusing on the case $m=1$ in a first reading.}.

Verification is reduced to checking positivity of a certain number of
polynomials in $\mathbf Q(\lambda)[n]$ for $n\ge N$, where $\lambda$
is the dominant eigenvalue of~$A$. Let~$e$ be a positive eigenvector
of~$A$ for~$\lambda$, assume that $v=Te$ is positive and consider two
convex cones pointed at~0. The first one is
\begin{equation}\label{eq:Brv} 
B_r(v)=\{x\in\mathbb R^d_{>0}\mid x_iv_j\le rx_jv_i\text{ for all }i,j\}.
\end{equation}
If $r=\infty$, this cone is~$\mathbb R^d_{>0}$. Otherwise, it is
generated by~$2^d-2$ vectors obtained by
choosing the $i$th coordinate in $\{v_i,rv_i\}$ so that the result is
neither~$v$ nor~$rv$. The second cone is its image
$$
C_r(v)=T^{-1}B_r(v).
$$
Verification proceeds in three steps. We first present it when~$m=1$:
\begin{itemize}
\item[]\emph{Sanity checks:} check $\lambda>0$, $v>0$, $C_r(v)\subset\mathbb R_{>0}^d.$
\item[]\emph{Initialization:} check that $U_n\ge0$ for $n<N$; check
that $U_N\in C_r(v)$.
\item[]\emph{Induction step:} check that $A(n)C_r(v)\subset C_r(v)$ for $n\ge N$.
\end{itemize}
When these steps are completed, it follows that for all $n\ge N$, one has $U_n\in C_r(v)\subset\mathbb R_{>0}^d$: positivity is proved. This induction effectively proves $d(d-1)$ linear inequalities on $(u_n)_{n\in\mathbb N}$ simultaneously, originating in the inequalities that define $B_r(v)$ (and only $d$ inequalities when $r=\infty$).

If $m>1$, the initialization also checks that $U_{N+1},\dots, U_{N+m-1}$ belong to~$C_r(v)$ and the induction step checks that 
$A(n+m-1)\dotsm A(n)C_r(v)\subset C_r(v)$
instead of $A(n)C_r(v)\subset C_r(v)$. The same argument shows that this proves positivity by induction.

\subsection{Complexity questions}
In terms of algorithmic complexity, there are two expensive steps: one related to the recurrence and another one related to the initial conditions. The induction step can be performed by checking that each of the $2^d-2$ vectors generating the cone~$C_r(v)$ (resp. $d$ vectors when $r=\infty$) has for image by~$A(n)$ a vector of rational functions that satisfies the $d(d-1)$  inequalities (resp. $d$ inequalities) defining the cone. This amounts to $d(d-1)(2^d-2)$ (resp. $d^2$) polynomials in~$\mathbb Q[n]$ that have to be proved positive for~$n\ge N$ (e.g., by Sturm sequences, or simply by certified numerical evaluation of the roots). The complexity of that step is thus singly exponential in the order of the recurrence. 

Concerning the initial conditions, checking~$U_n\ge0$ for $n\le N$ has complexity that is clearly polynomial in~$N$ (one can also use multipoint polynomial evaluation to reduce further the cost by evaluating the coefficients~$p_i(n)$ for $n\le N$ efficiently); this has complexity singly exponential in the \emph{bit size} of~$N$~\cite[Prop.~15.1]{BostanChyzakGiustiLebretonLecerfSalvySchost2017}.
Still, at the moment we do not have an upper bound on~$N$ in terms of the input, in particular in relation to a distance of the vector of initial conditions to the hyperplane of non-genericity.

\subsection{Two examples of certificate verification}
\label{subsec:example}
\begin{figure}
    \centering
     {\includegraphics[height=5cm, trim=100 100 100 100,clip=true]{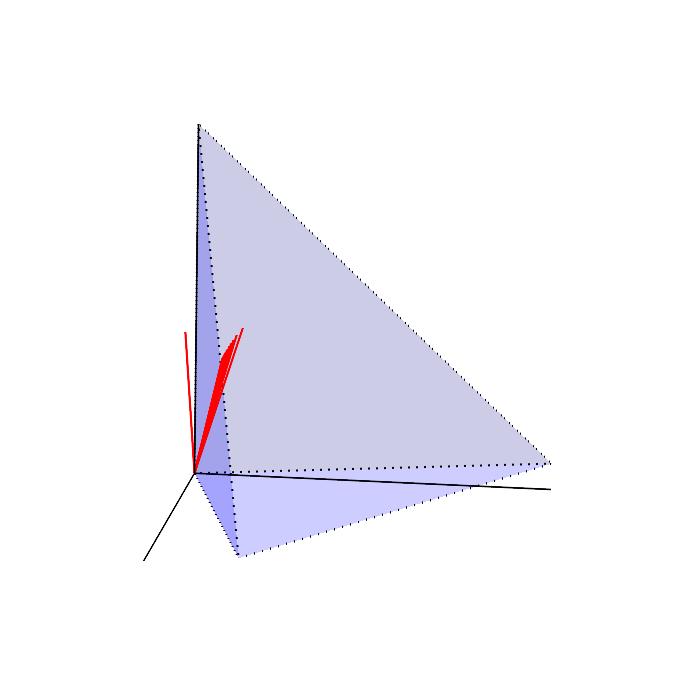}}\qquad\qquad{\includegraphics[height=5cm]{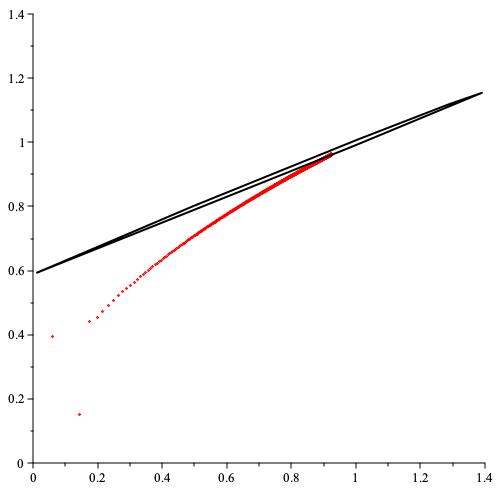}}
    \caption{The first values of $U_n$ in the examples of \cref{subsec:example} (in red), together with the corresponding cones $T^{-1}(B_r(v))$ (in black or blue). In both cases, the dimension is decreased by scaling the vectors $U_n$ so that their last coordinate is~1 and by taking the intersection of the cone with the hyperplane setting the last coordinate to~1.
     \label{fig:my_label}
}
\end{figure}

\begin{Example}We start with an example where~$r=+\infty$, where verification is easier.

The sequence defined by
\[u_{n}:=\sum_{k=0}^n (-1)^k \frac{(4n-3k)!(4!)^k}{(n-k)!^4k!}\]
is the first of a family related to a former conjecture of Gillis,
Reznick and Zeilberger~\cite{gillis1983elementary,Yu2019}. Its positivity was
proved automatically by Kauers~\cite{Kauers2007a} (see also~\cite{Pillwein2019}),
using the linear recurrence of order~4 that can be computed by Zeilberger's algorithm:
\begin{multline*}
(2n + 5)(4n + 11)(4n + 7)(n + 4)^3u_{n+4}
- 8(4n + 7)(4n + 13)(n + 3)(40n^3 + 380n^2 + 1193n + 1240)u_{n+3}\\
+576(192n^6 + 3072n^5 + 20108n^4 + 68918n^3 + 130513n^2 + 129613n + 52815)
u_{n+2}\\
+ 13824(4n + 15)(32n^5 + 344n^4 + 1424n^3 + 2855n^2 + 2801n + 1085)u_{n+1}\\
+ 331776(4n + 15)(4n + 11)(2n + 7)(n + 1)^3u_n=0.
\end{multline*} 

We give an alternate proof using a relatively
small certificate of positivity:
\[T=\begin{pmatrix}1& 0& 0 & 0\\
-1& 1& 0&0\\
0& -2& 1& 0\\
-3000&-1000&-40&1
\end{pmatrix},\quad r=+\infty,\quad N=3,\quad m=1.\]
This is illustrated in \cref{fig:my_label}.
The verification of this certificate thus consists in a proof by induction that the following inequalities are all satisfied for $n\ge3$:
\[
    u_n>0,\quad
    u_{n+1}> u_n,\quad
    u_{n+2}> 2u_{n+1},\quad
    u_{n+3}> 40 u_{n+1}+1000u_{n+2}+3000 u_{n+3}.
\]
We now turn to the verification.

The characteristic polynomial has one dominant
root~$\lambda\approx130$, of much larger modulus than the other ones.
The corresponding eigenvector $v=\trsp{
(1,\lambda,\lambda^2,\lambda^3)}$ is also positive. As $T^{-1}$ is a triangular matrix with positive elements below the diagonal, we get that $C_r(v)=T^{-1}\mathbb{R}_{>0}^d>0$, which concludes the `Sanity checks'.

Checking that the first~4 vectors~$U_0,U_1,U_2,U_3$ are positive is done by checking $u_i>0$ for $i=0,\dots,6$. With $U_3=\trsp{(18816, 1785816, 177396480, 18271143360)}$, it is easy to check that $TU_3>0$, i.e., $TU_3\in B_r(v)$ or equivalently $U_3\in C_r(v)$, concluding the initialization step.

Finally, as the cone $B_r(v)$ is $\mathbb R_{>0}^4$, the induction
step, which consists in checking that  $TA_nT^{-1}\mathbb R_{>0}^4\subset \mathbb R_{>0}^4$ for $n\ge3$ is readily achieved by a direct computation of $TA(n)T^{-1}$, which has the form
\[TA(n)T^{-1}=\begin{pmatrix}
1 & 1& 0&0\\
1 &1&1&0\\
4076&1076&38&1\\
a_1(n)&a_2(n)&a_3(n)&a_4(n)
\end{pmatrix},\]
with $a_i(n)$ rational functions. For instance, $a_1(n)$ is
\[
\frac{
8(362464n^6 + 12010912n^5 + 121406462n^4 + 567578151n^3 + 1363921108n^2 + 1636882352n + 779476880)}{(32n^6 + 608n^5 + 4738n^4 + 19353n^3 + 43628n^2 + 51376n + 24640)}
\]
making its positivity apparent. The same is true for the other ones. Therefore, the image of any vector with positive coordinates also has positive coordinates.

\end{Example}
\begin{Example}\label{example-6}
As an example with a finite $r$, we consider the sequence $(u_n)_{n\in\mathbb{N}}$ defined by the third-order
recurrence
$$(n + 1)u_{n + 3} = \left(\frac{77}{30}n + 2\right)u_{n + 2} - \left(\frac{13}{6} n - 3\right)u_{n + 1} +\left(\frac{3}{5}n + 2\right)u_{n},\qquad n\in\mathbb{N}, $$
with initial conditions $u_0=1$, $u_1=15/14$, $u_2=8/7$. This is a recurrence that falls outside of the domain reachable by the methods of Kauers and Pillwein~\cite{KauersPillwein2010,Pillwein2019}.
Here is a certificate for its positivity:
$$T=\begin{pmatrix}-36/7& 76/7& -33/7\\
162/7& -405/7& 250/7\\
303/14 & -4783/84& 3049/84
\end{pmatrix},\quad r=5/3,\quad N=3040,\quad m=1.$$

The dominant eigenvalue is~$\lambda=1$ and the vector~$v$ is~$\trsp{(1,1,1)}$. So the first part of the `Sanity checks' is easy.
The $2^d-2=6$ edge vectors of the cone $B_r(v)$ are
\[\begin{pmatrix}r\\
1\\
1\end{pmatrix},\quad\begin{pmatrix}1\\
r\\
1\end{pmatrix},\quad\begin{pmatrix}1\\
1\\
r\end{pmatrix},\quad\begin{pmatrix}r\\
r\\
1\end{pmatrix},\quad\begin{pmatrix}1\\
r\\
r\end{pmatrix},\quad\begin{pmatrix}r\\
1\\
r \end{pmatrix},\quad\text{with $r=5/3$}.\]
In order to check that $C_r(v)\subset\mathbb{R}^d_{>0}$, it is sufficient to test that $T^{-1}V>0$
for each of these vectors, concluding the `Sanity checks'.

For the initialization, one checks the positivity of the first $N$ terms of the sequence and that $TU_N$ satisfies the $d(d-1)=6$ inequalities that define~$B_r(v)$.

Finally, for the induction step, for each generator~$V$ of the cone~$B_r(v)$,
one checks that the vector of polynomials\footnote{More generally, one would check the polynomials $p_d(n)TA(n)T^{-1}V$ for $n$ such that $p_d(n)>0$.} $(n+1)TA(n)T^{-1}V$
satisfies the linear inequalities that define~$B_r(v)$.  For instance, for the generator $V_1=\begin{pmatrix}r&1&1
\end{pmatrix}\trsp{}$
one gets the polynomials
\begin{multline*}
360612n + 392450160, 1939140n - 264007440, 1247967n + 399271660,\\
1406727n - 268100340, 1839915n + 153100060, 420147n + 142185660,
\end{multline*}
that have to be proved positive for $n\ge N$. As they are all linear in~$n$ in this example, this is straightforward.

In terms of the sequence $(u_n)_{n\in\mathbb{N}}$, this proof shows by induction that for $n\ge n_0$, the following six inequalities are satisfied
\begin{align*}
0 &< -666u_n + 1595u_{n + 1} - 915u_{n + 2},&\quad 0 &< 918u_{n} - 2253u_{n + 1} + 1349u_{n + 2},\\ 
0 &< -2538u_{n} + 6303u_{n + 1} - 3709u_{n + 2},&\quad 0& < 3258u_{n} - 9335u_{n + 1} + 6245u_{n + 2},\\ 
0 &< 1422u_{n} - 3317u_{n + 1} + 1951u_{n + 2},&\quad 0 &< 10386u_{n} - 26651u_{n + 1} + 16433u_{n + 2}.
\end{align*}
\end{Example}

\section{Convergent contractions}\label{sec:Friedland}
The geometric insight on the convergence in Friedland's theorem makes use of Hilbert's pseudo-metric.
\begin{Definition} Hilbert's pseudo-metric on \(\mathbb R_{>0}^d\) is defined
by
\[d_H(x,y)=\log\frac{\max_i(x_i/y_i)}{\min_i(x_i/y_i)}.\]
\end{Definition}
Being a pseudo-metric means that \(d_H(x,x)=0\), \(d_H(x,y)=d_H(y,x)\)
and \(d_H(x,y)\le d_H(x,z)+d_H(y,z).\) All are easy to check. Moreover,
\(d_H(x,y)=0\) if and only if there exists \(\alpha>0\) such that
\(x=\alpha y\).

For this pseudo-metric, the closed ball centered at \(v\) and of
radius \(\log r\) is the cone \(B_r(v)\) from \cref{eq:Brv}.

\begin{theorem}[Birkhoff \cite{Birkhoff1957}] For a positive matrix
\(A\in\mathbb R_{>0}^{d\times d}\), let
\(L(A)=\sup_{x\neq \alpha y}d_H(Ax,Ay)/d_H(x,y)\). Then
\[L(A)=\frac{1-\sqrt{\psi(A)}}{1+\sqrt{\psi(A)}}\quad\text{with}\quad \psi(A)=\min_{i,j,k,\ell}\frac{a_{ik}a_{j\ell}}{a_{i\ell}a_{jk}},\]
showing that \(L\) is continuous and that \(A\) is a contraction.
\end{theorem}
This was used by Birkhoff to give a new proof of Perron's
theorem that
any positive matrix admits a unique positive eigenvector~\cite[Vol.~2,XIII,\S2]{Gantmacher1959} (and a
generalization in arbitrary dimension).

The key result for our method is the following theorem at the heart of Friedland's proof, that we make explicit for later use.
\begin{theorem}\label{thm:key} If \(A(n)>0\) tends to \(A>0\) as
$n\rightarrow\infty$, let $\lambda>0$ be the positive real eigenvalue
of~$A$ and
\(v>0\) be such that
\(Av=\lambda v\) ($\lambda$ and $v$ exist by Perron's theorem), then
for \(n\) sufficiently
large, \(A(n)B_r(v)\subset B_r(v)\).
\end{theorem}
\begin{proof}
Let \(x\in B_r(v)\), then
\[d_H(A(n)x,v)\le d_H(A(n)x,A(n)v)+d_H(A(n)v,v)
\le L(A(n))\log r+d_H(A(n)v,Av).\]
The first summand tends to \(L(A)\log r<\log r\), the second one to~0,
so the sum is smaller than \(\log r\) for \(n\) sufficiently large,
i.e., \(A(n)x\in B_r(v).\)
\end{proof}

Reduction to the positive case is achieved by the following.
\begin{lemma}[Friedland]\label{lemma:Friedland} For a matrix
$A\in\mathbf Q^{d\times d}$ with a simple dominant
eigenvalue~$\lambda>0$, there
exists \(T\in\operatorname{GL}_d(\mathbb Q)\)
such that \(TAT^{-1}\) has positive right and left eigenvectors \(a,b\)
for~\(\lambda\).
\end{lemma}
Note that in these conditions, \(TA^mT^{-1}/\lambda^m\) tends to
\(a\trsp{b}\) (by \cref{thm:Friedland}) and thus has to be
positive for some finite \(m\). This is where the last part~$m$ of our certificates comes from. As soon as the dimension~$d\ge3$, there are matrices~$A$ for which it is not possible to find a matrix~$T$ such that $TAT^{-1}>0$.
\begin{proof} Friedland's proof of the lemma is constructive (and leaves a lot of freedom in the construction of \(T\)). We reproduce it here to make the algorithmic part of this work self-contained.

Let $\lambda$ be the dominant eigenvalue of~$A$. There exists~$Q\in\operatorname{GL}_d(\mathbf Q(\lambda))$ such that~$B=QAQ^{-1}=(\lambda)\oplus B'$ for some~$B'\in\mathbf{Q}(\lambda)^{(d-1)\times(d-1)}$. If~$e_1$ denotes the vector~$\trsp{(1,0,\dots,0)}$, then $Be_1=\trsp{B}e_1=\lambda e_1$. Choose $a=\trsp{(1,\dots,1)}$ and $b$ a positive vector with coordinates in~$\mathbb{Q}$ such that $\trsp{a}b=1$. Let $(s_2,\dots,s_d)$ be a basis of vectors, all in~$\mathbb{Q}^d$ and orthogonal to~$b$ and form~$S\in\mathbb Q^{d\times d}$, the matrix with columns~$(a,s_2,\dots,s_d)$ so that $Se_1=a$ and $\trsp{S}b=e_1$. Let next $T=SQ$ and $M=TAT^{-1}$. Then,
\begin{align*}
Ma&=SBS^{-1}a=SBe_1=\lambda Se_1=\lambda a,\\
\trsp{M}b&=\trsp{(S^{-1})}\trsp{B}\trsp{S}b=\trsp{(S^{-1})}\trsp{B}e_1=\lambda\trsp{(S^{-1})}e_1=\lambda b.
\end{align*}
By continuity and density of~$\mathbb Q$, one can further restrict to~$T\in\operatorname{GL}_d(\mathbb Q)$.
\end{proof}

\section{Algorithm and proof}\label{sec:algo}
\begin{algorithm}
    \KwIn{A recurrence of Poincaré type, in the form of a matrix $A
    (n)\in\mathbf Q(n)^{d\times d}$; a vector $U_0\in\mathbf Q^d$ of
    initial conditions. It is assumed that $A=\lim_{n\rightarrow\infty}A(n)$ has a unique simple dominant eigenvalue $\lambda$ with  eigenvector~$e>0$. (This can be checked algorithmically.)}
    \KwOut{One of (Positive,$T,r,N,m$), (Non-positive)}  
    \BlankLine
    \If{$\lambda>0$}{
    Find $T\in\mathrm{GL}_d(\mathbb Q)$ and $m>0\in
      \mathbb{N}$
     such that $TA^mT^{-1}>0$ \;
    $v \gets T{e}$;
      \lIf{$v<0$}{ {$T\gets -T$,~ $v\gets -v$ }}
      Find $r>1$ such that $T^{-1}B_r(v)>0$\ {// $B_r(v)$ from \cref{eq:Brv}}\;
     Find $K\geq 0\in\mathbb{N}$ such that $n\geq K\Rightarrow TA({n+m-1})\dotsm A({n})T^{-1}(B_r(v))\subset B_r(v)$\;
    \For{$i=0,\dots,K-1$}{\lIf{$U_i\not\ge0$}{\KwRet(Non-positive)}}
    \For{$i=K,K+1,\dots,\infty$}{
        \lIf{$U_i\not\ge0$}{\KwRet (Non-positive)}
        \If{$TU_j\in B_r(v)$ for $j=K,\dots,K+m-1$}{\KwRet 
        (Positive,$T,r,m,i$)}
        }}
    \Else{
        \For{$i=0,\dotsc$}{\lIf{$U_i\not\ge0$}{\KwRet(Non-positive)}}
    }
    \caption{\PP}\label{PP}
\end{algorithm}



Algorithm \PP\ is a direct consequence of the results of the previous section. We now prove its correctness, thereby proving \cref{theorem:main}.

By Friedland's theorem, for generic initial conditions, the direction of $U_n$ tends to that of the
eigenvector corresponding to the unique dominant eigenvalue $\lambda$
of~$A$, which, being unique, is real. If $\lambda$ is negative, then
for sufficiently large~$n$, one of~$U_n$ and~$A(n)U_n$ has a negative coordinate. This is checked by~Step~14.

When $\lambda>0$, the next step is to compute a  $T\in\operatorname{GL}_d(\mathbb Q)$ and an integer $m>0$ such that $TA^mT^{-1}>0$. This
is possible by \cref{lemma:Friedland}.

Next, by definition of~$e$ and~$v$, $TA^mT^{-1}v=\lambda^mv$. So $v$ is an eigenvector for the positive eigenvalue of a positive matrix. By Perron's theorem, it is a real multiple of a positive vector. Thus either~$v>0$ or $v<0$ and then changing~$T$ into~$-T$ and~$v=Te$ into~$-v$ turns~$v$ into a positive eigenvector of~$TA^mT^{-1}.$ This is what is done in Step~3.

Since $e=T^{-1}v>0$, by continuity of the linear map $T^{-1}$, for sufficiently small $r>0,~T^{-1}B_r(v) >0$.
Such an~$r$ can be computed for instance by starting from~$r=2$ and using dichotomy to divide the distance between~$r$ and~$1$ until an appropriate~$r$ is found. 
This proves that Step~4 succeeds. 

By \cref{thm:key} applied to $TA(n+m-1)\dotsm A(n)T^{-1}$, there exists a~$K$ as required by Step~5. In order to compute it,
one can compute $TA({n+m})\dots A(n)T^{-1}G$ for each generator~$G$ of
the cone~$B_r(v)$, which gives a vector of polynomials in~$\mathbf Q
(\lambda)[n]$ that has to be ultimately positive by the existence
of~$K$. For instance, for each polynomial, one can start from~$i=1$ and check whether the polynomial is positive on~$[i,\infty)$ using Sturm sequences, and if not, double~$i$. In the end, $K$ can be taken as the maximum of the values obtained for each coordinate for each generator~$G$. This proves that Step~5 succeeds.

Step~6 consists simply in checking that the initial values up to~$K$ are nonnegative.

Finally, the termination of Steps~8--11 relies on Friedland's \cref{thm:Friedland}, which shows that for any generic vector of initial conditions, the direction of~$U_n$ tends to that of~$e$ and therefore $d_H(TU_n,v)\rightarrow0$ as $n\rightarrow\infty$. Thus for large enough~$j$, all $TU_j$ belong to~$B_r(v)$, showing that the seemingly infinite loop always terminates for generic initial conditions and concluding the proof.

\subsection*{Example}
For the sequence $(s_n)_{n\in\mathbb{N}}$ from \cref{eq:Straub} in the
introduction, denoting by $\mathrm{S}_n$ the vector $\trsp{(s_{n},s_
{n+1})}$, the recurrence is $\mathrm{S}_{n+1}=A(n)\mathrm{S}_{n}$ with 
\[A(n)=\begin{pmatrix} 0 &1\\
\frac{-81(3n +2)(3n + 4)}{2(n + 2)^2}&\frac{(81n^2 + 243n + 186)}{2(n + 2
)^2}
\end{pmatrix}\qquad\text{and}\qquad \mathrm{S}_0=
\begin{pmatrix} 
1\\ 12
\end{pmatrix}.\]
The limit matrix $A$ of $A(n)$ has one simple dominant eigenvalue $\lambda=27>0$ with $e=(1,\lambda)$ its associated eigenvector. We follow the steps of the algorithm.

First, following the steps in the proof of \cref{lemma:Friedland}, a possible choice of matrix is
\[T=\frac{1}{13} \begin{pmatrix}-14&1
\\
1&0
\end{pmatrix}.\]
Since  $TAT^{-1}>0$, we have $m=1$ and note that the vector $v=Te=(1,1)>0$.

Next, as the inverse of $T$ is triangular with positive elements below
the anti-diagonal, for all real $r>0$,  $T^{-1}B_r(v)>0$, showing that we can take $r=+\infty$. In Step 5, $K$ can be chosen as the rank for which the matrix  $TA(n)T^{-1}$ becomes positive. The value of this matrix is
\[
\begin{pmatrix}
\frac{53n^2+131n+74}{2(n+1)^2} & \frac{13n^2+376n+388}{26(n+1)^2}\\
13 & 14
\end{pmatrix},\]
showing that $K=0$ works.
After checking that $U_0$, $U_1$ and $TU_1$ are positive, the positivity is concluded and the final step of the algorithm finds $N=1$.

Note that this choice of matrix~$T$ means that the algorithm proves the positivity of $(s_n)_{n\in\mathbb{N}}$ by synthesizing and proving the  inequalities
\[s_{n+1}>14s_{n}>0,\qquad\text{for $n\ge1$.}\]
In this example, this
recovers the stragegy of Kauers and Pillwein~\cite{KauersPillwein2010}
of looking for an inequality~$u_{n+1}\ge\mu u_n$. One way of seeing
the improvement brought by our algorithm is that it will always succeed in producing a matrix~$T$ when the conditions of \cref{theorem:main} are met, while the inequalities~$u_{n+1}\ge\mu u_n$ correspond to a restricted set of band matrices. 

\section{Conclusion and Future Works}
Informally speaking, this work shows that positivity certificates can be computed for
a large class of P-finite sequences, whose positivity follows `in an
easy way' from their asymptotic behaviour. These certificates can be
viewed as a finite set of linear inequalities satisfied by $u_n,u_
{n+1},\dots,u_{n+d}$, whose simultaneous proof by induction implies the positivity of~$(u_n)_{n\in\mathbb N}$ and reduces to checking the positivity of a finite number of univariate polynomials.

We wrote a prototype implementation in Maple. 
In the current implementation, the choice of the matrix~$T$ is performed following the proof of \cref{lemma:Friedland}. This is far from optimal. The more difficult examples of this article use a matrix~$T$ obtained using several heuristics that have not been implemented yet. In practice, the choice of~$T$ has a strong impact on the value of the number~$N$ of terms that have to be tested positive. In the first example of \cref{subsec:example}, which has order~4, the matrix~$T$ we give leads to~$m=1$ and~$N=3$, so that checking the certificate is easy. By contrast, currently our implementation gives~$m=5$ but, more importantly, the large value $N=2420861$ which prevents checking the certificate in a reasonable time.

The constraint that the matrix~$A$ in \cref{theorem:main} has a unique
dominant eigenvalue that is simple is a limitation of our approach.
For instance, the sequence proved positive by Melczer and Mezzarobba~%
\cite{MelczerMezzarobba2022} by a purely analytic method has a
dominant eigenvalue
that is double and thus inaccessible by the approach presented here. There are cases where our approach seems to extend to a unique  dominant eigenvalue that is not simple. Also, the analytic approach cannot cope with parameters while our method seems to extend to this situation, at least in simple cases. We plan to explore this further. 

The condition of genericity of the initial conditions in \cref{theorem:main} and \cref{cor:recurrences} ensures that the solution is not subdominant and thus behaves asymptotically like the largest one, which is unique by our assumption on the dominant root of the characteristic polynomial.
For recurrences of order~2, subdominance is equivalent to minimality. Exploiting the relation between second order linear recurrences and continued fractions, Kenison \emph{et al.}~\cite{KenisonKlurmanLefaucheuxLucaMoreeOuaknineWhitelandWorrell2021} obtain a Turing reduction of positivity to minimality for the order~2 in two steps: first, they give an algorithm that terminates except in the minimal situation; next, in that non-generic situation, they provide a decision method via continued fractions. For recurrences of higher order, our work generalizes only the first part: if the recurrence satisfies the conditions of \cref{cor:recurrences} and in particular the characteristic polynomial has a unique dominant root, then our algorithm answers the positivity question but may fail to terminate in the non-generic situation. However, this does not constitute a reduction to the problem of deciding genericity of the initial conditions, as we do not know how to decide positivity in the non-generic situation in all generality. One possible direction of extension would be deal with special situations first, as for instance when all roots of the characteristic polynomial have distinct absolute values.

\bigskip

\paragraph{\bf Acknowledgements.} 
Alin Bostan and Mohab Safey El Din made many very useful suggestions
on previous versions of this article. The presentation also benefited from the feedback of the participants to the workshop \emph{Algorithmic Aspects of Dynamical Systems} at McGill University’s Bellairs Research Institute. We are also thankful to the referees whose comments helped us clarify some of our statements.

This work has been supported in part by the ANR project NuSCAP ANR-20-CE48-0014.

\printbibliography

\end{document}